\documentclass[twoside,leqno,twocolumn]{article}
\usepackage{etoolbox}
\newbool{anon}
\setbool{anon}{true}

\usepackage[table]{xcolor}
\usepackage[letterpaper]{geometry}

\usepackage{siamproceedings}
\usepackage[T1]{fontenc}
\usepackage{amsfonts}
\usepackage{graphicx}
\usepackage{epstopdf}
\usepackage{enumitem}
\usepackage{siunitx}
\usepackage{ifthen}

\ifpdf
  \DeclareGraphicsExtensions{.eps,.pdf,.png,.jpg}
\else
  \DeclareGraphicsExtensions{.eps}
\fi

\usepackage{amsopn}
\usepackage{fixmath}
\usepackage{mathtools}
\usepackage{xspace}

\usepackage{tikz}
\usetikzlibrary{graphs}
\usetikzlibrary {graphs.standard}
\usetikzlibrary{shapes.symbols}
\usepackage{intcalc}
\usepackage[format=plain,
            labelfont={bf,it},
            textfont=it]{caption}
\usepackage[ruled,vlined,linesnumbered,algo2e]{algorithm2e}

 \newcommand{\newindthm}[2]{
   \theoremstyle{plain}
   \theoremheaderfont{\normalfont\sc}
   \theorembodyfont{\normalfont\itshape}
   \theoremseparator{.}
   \theoremsymbol{}
   \newtheorem{#1}{#2}
 }
\newindthm{rgl}{Rule}
\newindthm{observation}{Observation}

\newsiamremark{remark}{Remark}
\newsiamremark{hypothesis}{Hypothesis}
\newsiamthm{claim}{Claim}

\crefname{hypothesis}{Hypothesis}{Hypotheses}
\crefname{claim}{Claim}{Claims}
\crefname{remark}{Remark}{Remarks}
\Crefname{hypothesis}{Hypothesis}{Hypotheses}
\Crefname{rule}{Rule}{Rules}
\crefname{rule}{Rule}{Rules}

\Crefname{claim}{Claim}{Claims}
\Crefname{remark}{Remark}{Remarks}

\crefname{subsection}{Subsection}{Subsections}
\crefname{paragraph}{Paragraph}{Paragraphs}
\crefname{section}{Section}{Sections}

\crefname{algocf}{Algorithm}{Algorithms}
\Crefname{algocf}{Algorithm}{Algorithms}

\def\Rule#1{\texorpdfstring{\textsc{Rule~#1}}{Rule~#1}\xspace}

\newcommand{\Dom}{\textsc{Dominating Set}\xspace}

\newcommand{\domset}{\ensuremath{\mathcal D}\xspace}

\newcommand{\tone}{type~\texorpdfstring{$1$}{1}\xspace}
\newcommand{\ttwo}{type~\texorpdfstring{$2$}{2}\xspace}
\newcommand{\tthree}{type~\texorpdfstring{$3$}{3}\xspace}
\newcommand{\ttwothree}{type~\texorpdfstring{$2$}{2} and~\texorpdfstring{$3$}{3}\xspace}
\newcommand{\ie}{i.e.,\xspace}
\newcommand{\eg}{e.g.\xspace}
\newcommand{\etal}{et al.\xspace}

\colorlet{ref}{black}
\colorlet{wit}{black}

\def\refa{\ensuremath{{\color{ref}\rho}}\xspace}
\def\refb{\ensuremath{{\color{ref}\sigma}}\xspace}
\def\nodeu{\ensuremath{{u}}\xspace}
\def\nodev{\ensuremath{{v}}\xspace}
\def\witu{\ensuremath{{\color{wit}u}}\xspace}

\def\Oh#1{\ensuremath{\mathcal O\!\left(#1\right)}}

\newcommand{\cand}[1][]{\ensuremath{\mathcal C_{#1}}\xspace}
\newcommand{\sel}[1][]{\ensuremath{\mathcal S_{#1}}\xspace}
\newcommand{\superS}[1][]{\texorpdfstring{\ensuremath{\mathcal{S}_{#1}^{2,3}}}{S23}\xspace}

\def\refset{\ensuremath{{\color{ref}\mathsf{ref}}}}
\def\witset{\ensuremath{{\color{wit}\mathsf{wit}}}}
\def\crepr{\ensuremath{\mathsf{canRef}}}
\def\isTypeAt#1#2#3{\ensuremath{({\color{wit}{#2}} \lhd_{#1} {\color{ref}#3})}}
\def\isThreeAt#1#2{\isTypeAt{3}{#1}{#2}}
\def\isTwoThreeAt#1#2{\isTypeAt{\set{2,3}}{#1}{#2}}

\DeclareMathOperator*{\argmax}{arg\,max}

\DontPrintSemicolon
\LinesNumbered
\SetKwBlock{AlgoDetails}{}{}
\SetInd{0.4em}{0.4em}

\newcommand{\setc}[2]{\ensuremath\left\{#1\;\middle|\;#2\right\}}
\newcommand{\set}[1]{\ensuremath\left\{#1\right\}}
\def\deg{\ensuremath{\mathrm{deg}}}

\def\nocomment#1{}

\newcommand{\ignore}[1]{}

\renewenvironment{equation*}{\[}{\]\ignorespacesafterend}

\begin{document}

\title{\Large Revisiting a Successful Reduction Rule for Dominating Set\footnote{
    Partially funded by the Deutsche Forschungsgemeinschaft (DFG) – ME 2088/5-2, FOR 2975.
}}
\author{
      Lukas Geis\thanks{Goethe University Frankfurt, Germany, \{lgeis,aleonhardt,umeyer,mpenschuck\}@ae.cs.uni-frankfurt.de}
      \and Alexander Leonhardt\footnotemark[2]
      \and Johannes Meintrup\thanks{University of Applied Sciences Mittelhessen, Germany, \email{johannes.meintrup@mni.thm.de}}
      \and Ulrich Meyer\footnotemark[2]
      \and Manuel Penschuck\footnotemark[2]
      \and Lukas Retschmeier\thanks{BARC, University of Copenhagen, Denmark, \email{lure@di.ku.dk}}
  }
  
\date{}
\twocolumn
\maketitle

\begin{abstract}
    Given a graph $G=(V,E)$ with $n$ vertices and $m$ edges, the \textsc{Dominating Set} problem asks for a set $\domset\subseteq V$ of minimal cardinality such that every vertex either is in $D$ or adjacent to a member of $D$.
    Although there is little hope for a kernelization algorithm on general graphs due to the $\textnormal{W}[2]$-hardness of \textsc{Dominating Set}, data reduction rules are extensively used in practice.
    
    In this context, \textsc{Rule} $1$ due to Alber, Fellows, and Niedermeier [JACM 2004] has been shown to be very powerful, yet its best-known running time is $\mathcal{O}(n^3)$ ($= \mathcal{O}(nm)$) for general graphs.
    
    In this work, we propose, to the best of our knowledge, the first $\mathcal{O}(n + m)$-time algorithm for \textsc{Rule} 1 on general graphs. 
    We additionally propose simple, but practically significant, extensions to our algorithmic framework to further prune the input instances.

     We complement our theoretical claims with experiments that confirm the practicality of our approach.
    On average, we see significant speedups of over one order of magnitude while removing $59.8\times$ more nodes and $410.9\times$ more edges than the original formulation across a large dataset comprised of real-world and synthetic networks.
    
~\\[0.3em]
\textsl{The source code and the experimental data is available at  \url{https://zenodo.org/records/17302472}\,.}
\end{abstract}

\fancyfoot[R]{\scriptsize{Copyright \textcopyright\ 2026 by SIAM\\
    Unauthorized reproduction of this article is prohibited}}

\section{Introduction}\label{sec:introduction}

Given a graph $G=(V,E)$ with $n$ vertices and $m$ edges, the \Dom problem asks for a set $\domset\subseteq V$ of minimal cardinality, such that every vertex either is in \domset or adjacent to a member of \domset.
\Dom was featured among Karp's~\cite{Karp72} original 21 NP-hard problems as a special case of \textsc{Set Cover} and stands as one of the most prominent problems of the parameterized complexity class~$\textnormal{W}[2]$ \cite{DowneyF95}.

Nonetheless, safe and fast data reduction rules are used extensively in practical solvers.
Early practical adoptions of domination related reduction rules date back to industrial applications in 1998~\cite{karsten1998}.
\Rule 1, initially designed as part of a kernelization routine for planar graphs, by Alber, Fellows and Niedermeier~\cite{DBLP:journals/jacm/AlberFN04} is among the reduction rules with the broadest impact.
Whenever a graph contains a vertex $u$ of degree one, its neighbor can be safely added to a minimal \Dom by a replacement argument; $u$'s neighbor is at least as good as adding $u$ to the dominating set.
Alber~\etal~\cite{DBLP:journals/jacm/AlberFN04} generalized this observation: if at least one node $u$ is highly \emph{confined} (formalized later in \cref{def:rule1}) in the neighborhood of some other node $v$, then we can remove a subset of the vertices in $N(v)$ and add $v$ to the \Dom.
In the two decades since its introduction, \Rule 1 has become a foundational building block in many algorithms and is applicable with minor modifications to a host of other (domination) related problems~\cite{retschmeier_masters_nodate,Garnero2018,GarneroST17,Luo2013}.

Although a fast $\Oh{n}$ time algorithm is known for planar graphs, on general graphs the best-known algorithm runs in time $\Oh{n^3}$ ($= \Oh{nm}$) \cite{DBLP:journals/jacm/AlberFN04} and it is an open question whether a better algorithm exists for all graphs.

\vspace{2mm}
\noindent\textbf{Our Contributions.}
We contribute the first algorithm for \Rule1 introduced by Alber~\etal~\cite{DBLP:journals/jacm/AlberFN04} that runs in time $\Oh{n + m}$ on \emph{all} graphs, prove that instead of repeated applications, a single application of the original rule suffices, and support the claimed practicability of our algorithm with comprehensive experiments.

We also identify additional reduction opportunities by extending  \cref{def:rule1} to include additional pruning steps.
These come at almost no cost since we effectively show in \cref{sec:applying} that a larger set of safely removable nodes and edges can be more easily identified than with a faithful implementation of the original \Rule 1.

We empirically demonstrate the effectiveness on large instances of \Dom, and evaluate its merits on a dataset with more than $\num{4500}$ graphs, comprised of synthetic and real-world instances with up to $113$ million nodes and $1.8$ billion edges.
On this dataset, we demonstrate that our most capable extension of \Rule1 removes, on average, $59.8\times$ more nodes and $410.9\times$ more edges compared to the original formulation.

At the same time, our implementation is $12.6\times$ faster on average, while for some instances we get a speedup of more than two orders of magnitude.
We scrutinize the comparatively small speedups gained given the asymptotic difference of our new algorithm in the experimental part and argue that, despite its cubic worst-case runtime, the original algorithm has a more favorable runtime complexity in expectation on a large class of graphs.

\vspace{2mm}
\textbf{Related Work.}
The reduction rules proposed by Alber~\etal have previously been shown to work well to significantly reduce (planar) instances of \Dom~\cite{Alber2006}.
Other practical work has implemented much simpler reduction rules~\cite{WangCCY18}, implemented (variants of) the reductions of Alber~\etal for small problem instances~\cite{JiangZ23,XiongX24}, or deemed a implementation of the reduction rules of Alber~\etal as computationally infeasible~\cite{AbuKhzamCESW17}.
The latter authors stated that they were only able to apply the reduction rules to tiny graphs, and for large (${\sim} 2000$ nodes) graphs the rules were simply too slow to be applied.
Prior work that specialized on solving massive \Dom instances only applied very simple reduction rules, \eg, removing nodes of degree at most $1$~\cite{WangCCY18}.
Recent state-of-the-art local search heuristics similarly rely on removing only nodes of degree at most $1$ and triangles~\cite{DBLP:journals/kbs/ZhuZWSL24}, both of which can be considered (very constrained) special cases of \Rule 1.

Fomin~\etal~\cite{FominLSZ19} showed that \Dom can be solved exactly in
$\Oh{b^n}$ for ${b=1.5137}$, which was subsequently improved to $b = 1.4969$ by Rooij~\etal~\cite{DBLP:journals/dam/RooijB11}. 
The currently best is $b = 1.4864$ by Iwata~\cite{Iwata2012}.

Regarding parameterized complexity, there exists a linear-kernel of size $67k$ for \Dom on planar graphs which was later improved to $43k$ \cite{Alber2006,Hagerup12, Halseth16}.
The central ideas in \cite{Alber2006} inspired follow-up works for linear kernels for various other hard problems on planar graphs, e.g. \{\textsc{total}, \textsc{semi-total}, \textsc{Red-Blue}, \textsc{Connected}, \textsc{Efficient}, \textsc{Edge},  \textsc{Directed}, \textsc{Vector}\} \Dom with kernel sizes ranging from $14k$ to $410k$. \cite{Garnero2018,GarneroST17,Hagerup12,LokshtanovMS11,Alber2006,retschmeier_masters_nodate,guo2007,Luo2013,sahili2025linearkernelplanarvector}.
Note that most of the works also rely on the reductions of Alber~\etal \cite{Alber2006}, either directly or with minimal problem-specific changes.

On a theoretical level, the techniques of Alber~\etal have been subsumed by the powerful machinery of so-called
\emph{protrusion decompositions} pioneered by Kim~\etal~\cite{KimLLRRSS15}, based on
finding a relatively small vertex set $X \subset V(G)$ of a graph $G$ called a
\emph{treewidth modulator}, that, when removed from $G$, turns $G$ into a graph
of constant treewidth. Using protrusion decompositions one can show a linear
sized, but wholly impractical, kernel for \Dom on planar graphs of
size $3\,146\,131\,968\,k$~\cite{FominLSZ19}. 

\section{Preliminaries}\label{sec:preliminaries}
Let $G=(V, E)$ be a graph with $n$ nodes and $m$ edges.
Throughout this document, we assume that:
all graphs are undirected, unweighted and simple (\ie have no multi-edges or loops), the set of nodes is totally ordered, and the nodes can index into arrays in constant time.
The last two assumptions are met, \eg by identifying $V = [n]$.
We denote the set of neighbors of node~$\nodeu$, its so-called (open) neighborhood, as $N(\nodeu)$ and call its size the \emph{degree} $\deg(\nodeu) = |N(\nodeu)|$.
The closed neighborhood is given by $N[\nodeu] = N(\nodeu) \cup \set{\nodeu}$.
Two nodes $\nodeu, \nodev \in V$ with $\nodeu \ne \nodev$ are called \emph{twins} if $N[\nodeu] = N[\nodev]$.
Furthermore, we denote 
the open neighborhood of a set of vertices $S\subseteq V$ as $N(S) \coloneqq \bigcup_{\nodeu\in S}N(\nodeu)$ and the closed neighborhood as $N[S] \coloneqq N(S) \cup S$.

\section{A New View on Rule 1}\label{sec:suitable_ref_nodes}
\begin{figure}
    \begin{center}
        \scalebox{1.0}{
        \begin{tikzpicture}[vertex/.style={inner sep=0, minimum width=1.5em, minimum height=1.5em, draw, circle}]
            \node[vertex] (n1-1) at (-2.75, 1.5) {a};
            \node[vertex] (n1-2) at (-2.25, 2) {b};

            \node[vertex] (n2-1) at (-.25, 1.5) {c};
            \node[vertex] (n2-2) at (.25, 2) {d};

            \node[vertex] (n3-1) at (2.75, 1.5) {e};
            \node[vertex] (n3-2) at (2.25, 2) {f};

            \node[vertex,thick,fill=gray!20] (ref) at (0, 0.7) {\refa};

            \path[draw, opacity=0.5, bend left] (ref) to (n1-1);
            \path[draw, opacity=0.5, bend left] (ref) to (n1-2);
            \path[draw, opacity=0.5, bend left] (ref) to (n2-1);
            \path[draw, opacity=0.5, bend right] (ref) to (n2-2);
            \path[draw, opacity=0.5, bend right] (ref) to (n3-1);
            \path[draw, opacity=0.5, bend right] (ref) to (n3-2);

            \path[draw, dotted] (n1-1) to ++(-1,0.25);
            \path[draw, dotted] (n1-1) to ++(-1,0);
            \path[draw, dotted] (n1-1) to ++(-1,-0.25);
            \path[draw, dotted] (n1-2) to ++(-1,0.25);
            \path[draw, dotted] (n1-2) to ++(-1,0);
            \path[draw, dotted] (n1-2) to ++(-1,-0.25);

            \path[draw] (n1-1) to (n2-1);
            \path[draw] (n1-1) to (n2-2);
            \path[draw] (n1-2) to (n2-2);

            \path[draw] (n2-2) to (n3-2);

            \path[dashed, draw] (-0.75-0.5, 0.3) to ++(0, 3.5);
            \path[dashed, draw] (0.75+0.5,  0.3) to ++(0, 3.5);

            \node[anchor=south, align=center] at (-2.5, 2.5) {\tone  \\ \footnotesize is \\[-0.5em] \footnotesize connected to \\[-0.5em] \footnotesize \phantom{y}the outside\phantom{1}};
            \node[anchor=south, align=center] at (0,    2.5) {\ttwo   \\ \footnotesize is \\[-0.5em] \footnotesize connected to \\[-0.5em] \footnotesize \tone};
            \node[anchor=south, align=center] at (2.5,  2.5) {\tthree \\ \footnotesize is \textbf{not} \\[-0.5em] \footnotesize connected to \\[-0.5em] \footnotesize \tone };

        \end{tikzpicture}}
    \end{center}

    \vspace{-1em}

    \caption{Neighbor types relative to reference node~$\refa$.
        As nodes $a$ and $b$ have neighbors outside of $N[\refa]$, they are of \tone.
        The remaining \ttwothree have only neighbors within $N[\refa]$, with the distinction
        that nodes $c$ and $d$ are connected to \tone nodes and thus are \ttwo.
        Since $\refa$ has at least one \tthree node, \Rule1 applies:
        We add $\rho$ to the dominating set.
        Then, it covers all its neighbors and there is no reason to add \ttwothree nodes into the dominating set.
        Hence, we can safely delete them.
        Observe that we keep \tone nodes, since they may be needed to cover the ``outside'' nodes.
    }
    \label{fig:nodetypes}
\end{figure}

One of the easiest reduction rules for \Dom is \Rule{Leaf}:
Consider a \emph{reference node}~\refa with at least one adjacent leaf~\nodeu (\ie a node with degree~1).
Since~\refa can only have a superset of neighbors $N[\nodeu]$, we say that \refa \emph{dominates} node~\nodeu.
Then, it is always safe to ignore the dominated leaf and instead add~\refa into the solution --- even if \refa is a leaf itself.

\Rule1 by Alber \etal~\cite{DBLP:journals/jacm/AlberFN04} generalizes this idea.
We formally state it in \cref{def:rule1}.
Again consider a fixed reference node\footnote{
    We denote exactly all reference nodes with greek letters.
}~$\refa \in V$.
Then, as illustrated in \cref{fig:nodetypes}, we partition its neighbors~$N(\refa)$ into three, possibly empty, types
--- ranging from interconnected (\tone), over weakly-dominated (\ttwo), to dominated (\tthree).
This type partition is always in terms of the fixed reference node; as sketched in~\cref{fig:different_types}, a node can have different types in terms of different reference nodes.
A \textbf{type 1} neighbor has at least one edge outside the closed neighborhood $N[\refa]$ of the reference node~\refa:
\begin{equation}
    N_1(\refa) \coloneqq \setc{\nodeu \in N(\refa)}{N(\nodeu) {\setminus} N[\refa] \ne \emptyset} \label{eq:one}\,.
\end{equation}

\noindent
A node is of \textbf{type 2}, iff it is not \tone, but connected to at least one \tone neighbor of \refa:
\begin{equation}
    N_2(\refa) \coloneqq \setc{\nodeu \in N(\refa) {\setminus} N_1(\refa)}{N(\nodeu) \cap N_1(\refa) \ne \emptyset}\label{eq:two}
\end{equation}

All other neighbors are of \textbf{type~3}.
They are only connected to~$\rho$ and possibly some of its \ttwo and \tthree neighbors.
As such, they are the most isolated type class containing, among others, the aforementioned leaves:
\begin{equation}
    N_3(\refa) \coloneqq N(\refa) \setminus N_1(\refa) \setminus N_2(\refa)\,.
\end{equation}

\noindent Based on these definitions, Alber~\etal~\cite{DBLP:journals/jacm/AlberFN04} show the following reduction rule (modified to match our notation):
\begin{rgl}\label{def:rule1}
    If a reference node~$\refa$ has $u \in N_3(\refa)$, we may modify the instance as follows:
    \begin{itemize}[noitemsep]
        \item add~$\refa$ to the dominating set, and
        \item delete all neighbors $N_2(\refa) \cup N_3(\refa)$.
    \end{itemize}
\end{rgl}

We say a node $\refa$ added to the dominating set by \Rule1 is \emph{fixed} (in the dominating set).
To prove that the rule is safe~\cite[Lemma~1]{DBLP:journals/jacm/AlberFN04} (\ie does not change the domination number of the instance if applied) it suffices to show that there always exists a minimal dominating set containing~\refa if \refa has at least one \tthree neighbor.
Then, if we assume that~\refa is part of the dominating set, it follows that the dominated neighbors~$N_2(\refa) \cup N_3(\refa)$ are redundant and can be deleted.

Observe that Alber~\etal implicitly add~\refa to the dominating set, by attaching a gadget to~\refa to force the reference node into the solution.
For the purpose of this article, it is an inconsequential implementation detail that we omit in the remainder (see also \cref{sec:applying}).

\subsection{Canonical reference nodes}\label{subsec:canonical-reference-nodes}
\begin{figure}
    \def\na{\textcolor{black!50}{$-$}}
    \begin{center}
        \scalebox{1.0}{
        \begin{tikzpicture}[
                vertex/.style={draw, circle, minimum width=1.6em, minimum height=1.6em, inner sep=0},
                edge/.style={draw},
            ]
            \node[vertex] (u0) at (3, -4) {$x$};
            \node[vertex, fill=black!10] (u1) at (-2, -2) {\small $\rho_1$};
            \node[vertex, fill=black!10] (u2) at (3, -2) {\small $\rho_3$};
            \node[vertex, very thick] (u5) at (0, -3) {$v$};
            \node[vertex, fill=black!10] (u6) at (0, -4) {\small $\rho_2$};
            \node[vertex, fill=black!10] (u7) at (1, -3.3) {$\rho_1'$};
            \path[edge] (u0) to (u2);
            \path[edge] (u1) to (u2);
            \path[edge] (u1) to (u5);
            \path[edge] (u1) to (u6);
            \path[edge] (u2) to (u5);
            \path[edge, bend left] (u2) to (u6);
            \path[edge] (u2) to (u7);
            \path[edge] (u5) to (u6);
            \path[edge] (u5) to (u7);
            \path[edge] (u6) to (u7);
        \end{tikzpicture}}

        \medskip

        \scalebox{1.0}{\begin{tabular}{c|c|c|c|c|c|c}
            $\nearrow$ is type $i$ neighbor of & $\rho_1$ & $\rho_1'$ & $\rho_2$ & $\rho_3$ & $v$ & $x$ \\\hline\hline
            $\rho_1$                            & \na      & \na       & 2        & 3        & 2   & \na \\\hline
            $\rho_1'$                           & \na      & \na       & 2        & 3        & 2   & \na \\\hline
            $\rho_2$                            & 1        & 1         & \na      & 3        & 2   & \na \\\hline
            $\rho_3$                            & 1        & 1         & 1        & \na      & 1   & 1   \\\hline
            \rowcolor{black!10}
            $v$                                 & 1        & 1         & 2        & 3        & \na & \na \\\hline
            $x$                                 & \na      & \na       & \na      & 3        & \na & \na \\
        \end{tabular}}
    \end{center}
    \caption{
        Node $v$ appears as different neighbor types for different reference nodes, namely:
$v \in N_1(\rho_1)$ as $\rho_1' \notin N(\rho_1)$. Furthermore, we have $v \in N_1(\rho'_1)$ as ${\rho_1 \notin N(\rho'_1)}$.
Furthermore, $v \in N_2(\rho_2)$ as $N[v] \subseteq N[\rho_2]$ and $v$ is connected to $\rho_3 \in N_1(\rho_2)$ which is a \tone neighbor to $\rho_2$.
Finally, we also have $v \in N_3(\rho_3)$ as $N[v] \subseteq N[\rho_3]$ and $N(v) \cap N_1(\rho_3) = \emptyset$.
    }
    \label{fig:different_types}
\end{figure}

The original \Rule1~\cite{DBLP:journals/jacm/AlberFN04} searches for a single suitable reference node in time $\Oh{n^3}$, applies the reduction of \cref{def:rule1}, and repeats to search for a new reference node until no more modification is possible.\footnote{
    There is no gain from applying the rule twice to the same node.
    Depending on implementation details, this corner case would need to be checked explicitly in the original formulation.
}
In the following, we derive a novel formulation to find \emph{all suitable} reference nodes and to apply the reduction in time $\Oh{n+m}$.

For clarity, we use the notation $\isTypeAt{T}{\nodeu}{\refa}$ to express ``\textsl{\nodeu is a type~$T$ neighbor of \refa}''.
We call node~\nodeu a \emph{witness} and \refa a \emph{reference node}, which we always denote by a greek letter.
The $\lhd$-symbol stylizes the domination order, pointing from the (typically) smaller neighborhood of $\nodeu$ to the larger of $\refa$.
This notation is only ``syntactic sugar'' to the reader.
We treat it synonymously to the typeless tuple $(\nodeu, \refa)$ and thus consider relations of different types on the same nodes as identical.
Then, we can define the candidate set~\cand of \emph{all possible} reference nodes along with their witnesses:
\begin{equation*}
    \cand = \setc{ \isThreeAt{\nodeu}{\refa} }{\refa \in V,\ \nodeu \in N_3(\refa)}\,.
\end{equation*}

For convenience, we define the following helper functions to extract the reference- and witness nodes, irrespective of the relations' types:
\begin{eqnarray*}
    \refset(X) &=& \setc{\refa}{ \isTypeAt{?}{\nodeu}{\refa} \in X},\\
    \witset(X) &=& \setc{\witu}{ \isTypeAt{?}{\witu}{\refa} \in X}\,.
\end{eqnarray*}

Observe that \Rule1 is inherently ambiguous.
For instance, in the complete graph $K_n=(V, E)$ for $n \ge 2$, every node~$\witu \in V$ is a \tthree neighbor for every other node~$\refa \in V \setminus \set{\nodeu}$.
Since we have $|\cand[K_n]| = \Theta(n^2)$, it is both expensive and meaningless to compute the set explicitly.
Furthermore, after a single application of the reduction rule, no witness nodes remain and thus every node has realistically only witnessed a single reference node.
Hence, we add a tie-breaker to determine a single reference node for every witness at the start. 
Then, a potential \tthree neighbor~$\witu$ may only be the witness to its canonical reference node~$\crepr(\nodeu)$.

The canonical reference node~$\crepr(\nodeu)$ of node~\nodeu is defined as ``the'' neighbor in $N[u]$ of the largest degree (including nodes outside~$N[u]$).
If there are multiple such neighbors, we chose the one of largest node index (recall that we assume $V$ to be totally ordered):

\begin{equation}
    \crepr(\nodeu) = \argmax_{\nodev \in N[\nodeu]}\set{ (\deg(\nodev), \nodev) }\,.
\end{equation}

Since $\crepr(\nodeu)$ is defined in terms of the closed neighborhood $N[\nodeu]$, a node can be its own reference node.
This prevents the node from becoming a witness, as no node can be a \tthree neighbor of itself.
Then, the resulting \emph{suitable set} $\sel \subseteq \cand$ is defined as
\begin{equation}\label{def:suitable-candidates}
    \sel = \setc{\isThreeAt{\nodeu}{\refa} \in \cand}{\refa = \crepr(\nodeu)}.
\end{equation}

The main insight is that the references nodes~$\refset(\sel)$ in the suitable set~\sel,
are ---up to isomorphism--- exactly the references found by the exhaustive application of the original \Rule1.
In the following, we show this in multiple steps.
We first show that the canonization $\sel \subseteq \cand$ omits only redundant pairs:

\begin{lemma} \label{lemma:canon}
    Each rejected candidate $\isThreeAt{\nodeu}{\refa} \in \cand \setminus \sel$ has an equivalent $\isThreeAt{\nodeu}{\refb} \in \sel$ such that $N[\refa] = N[\refb]$.
\end{lemma}
\begin{proof}
    By construction both reference nodes~\refa and~\refb dominate node~\nodeu, \ie
    \begin{align*}
        \left.\begin{array}{l}
                  \nodeu \in N_3(\refa) \Rightarrow N[\nodeu] \subseteq N[\refa] \\[0.1em]
                  \nodeu \in N_3(\refb) \Rightarrow N[\nodeu] \subseteq N[\refb]
              \end{array}\right\}\ \Rightarrow\
        \set{\refa, \refb} \subseteq N(\nodeu).
    \end{align*}

    Thus, we know that both reference nodes~\refa and~\refb are adjacent.
    Further, they cannot be \tone neighbors to each other, since node~\nodeu is a \tthree neighbor to both and therefore cannot be connected to \tone nodes:
    \begin{equation*}
        \underbrace{\left( N[\refb] {\subseteq} N[\refa] \right)}_{\text{as }\refb \not\in N_1(\refa)}
        \land
        \underbrace{\left( N[\refa] {\subseteq} N[\refb] \right)}_{\text{as }\refa \not\in N_1(\refb)}
        \ \ \Rightarrow\ \
        N[\refa] = N[\refb]
        \hfill\qed\,.
    \end{equation*}
\end{proof}

Recall that our ultimate goal is to apply \Rule1 of \cref{def:rule1}.
We now show that the reference nodes can be processed in any order (or conceptually even concurrently).
To this end, we have to convince ourselves that, after applying the reduction to any $\refa \in \refset(\sel)$, all remaining $(\refset(\sel) \setminus \set{\refa})$ remain legal targets for \Rule1.
Thus, we have to ensure that the rule's single precondition, namely that a target~$\sigma$ needs to have at least one \tthree node, remains intact.

\goodbreak

\begin{lemma} \label{lemma:unaffected}
    Assume $S$ as defined in Equation \ref{def:suitable-candidates}.
    Applying \Rule1 to any $\rho \in \refset(\sel)$, leaves the \tthree neighborhood of any other node $\refb \in \refset(\sel)$ with $\refa \ne \refb$ unaffected.
\end{lemma}
\begin{proof}
    We only need to consider the removal of \ttwothree nodes, as all other steps are trivially independent.
    Then by contradiction:

    Assume that by removing the \ttwothree neighborhood of~\refa, we also shrink the \tthree neighborhood $N_3(\refb)$ of another reference node $\refb \in (\refset(\sel) \setminus \set{\refa})$.
    Let \nodeu be a node inadvertently removed from $N_3(\refb)$ when processing~$\refa$.
    Since \nodeu was removed when deleting \ttwothree neighbors of~$\rho$, we know that
    \begin{equation*}
        \nodeu \in (N_2(\refa) \cup N_3(\refa)) \ \ \Rightarrow\ \ N[u] \subseteq N[\rho]\,.
    \end{equation*}

    Since (i) $u$ is \tthree to \refb and (ii) \refa and \nodeu are adjacent, we further know $\refa \in N_2(\refb) \cup N_3(\refb)$.
    Thus, we directly get $N(\refa) \subseteq N[\refb]$.

    Since $\refa \in \refset(\sel)$, there exists $\nodev \in N_3(\refa)$.
    Together with $N(\refa) \subseteq N[\refb]$, it follows that $v \in N[\refb]$.
    Hence, $\refb \in N_2(\refa) \cup N_3(\refa)$ (otherwise $\nodev$ were not a \tthree neighbor of \refa).
    Then, it directly follows that $N(\refb) \subseteq N[\refa]$.

    From $N(\refa) \subseteq N[\refb]$ and $N(\refb) \subseteq N[\refa]$, we get $N[\refa] = N[\refb]$.
    In other words, $\refa$ and $\refb$ are twins which contradicts the definition of $\sel$:
    Twins need to choose the same canonical representative (recall they can pick themselves and there is a tie-breaker based on the node index).
    Thus, we have $\set{\refa, \refb} \not\subseteq \refset(\sel)$.
\end{proof}

\def\exset{\ensuremath{\mathcal M}\xspace}
\begin{lemma} \label{lemma:maxcard}
    Consider the exhaustive application of \Rule1.
    Denote by $\exset \subseteq V$ the set of reference nodes it selected.
    Then $|\refset(\sel)| \geq |\exset|$.
\end{lemma}
\begin{proof}
    If we have $|\exset| = 0$, then no suitable target exists, \ie $|\cand| = 0$ which implies $|\sel|=0$.
    Thus, we assume $|\exset| > 0$ and iteratively peel away vertices until $\exset$ becomes empty.
    Let $\refa \in \exset$ be an arbitrary node.
    If $\refa \in \refset(\sel)$ apply \Rule1, thereby removing $\refa$ from $\exset$ and $\refset(\sel)$, and recurse.

    Thus, let us assume that $\refa \notin \refset(\sel)$.
    The only mechanism by which $\isThreeAt{\cdot}{\refa} \notin \sel$ while $\refa \in \exset$, is the canonization.
    From \cref{lemma:canon}, we know that $\refa \notin \refset(\sel)$ implies that there exists an equivalent twin $\refb \in \refset(\sel)$ with $N[\refa] = N[\refb]$.
    This implies that $\refb \notin \exset$ as it is a \ttwo node with respect to $\refa \in \exset$ and must be removed after applying \Rule1 to $\refa$.

    Since $\refa$ and $\refb$ are twins, applying \Rule1 to either produces equivalent results.
    Hence, we can apply \Rule1 to $\refa \in \exset$.
    Thereby, we remove $\refa$ from \exset, and can safely peel away the equivalent $\refb$ from $\sel$.
    As a result, both sets are reduced by one element each.
\end{proof}

\section{Finding All Suitable Reference Nodes}\label{sec:linear_time_search}
\def\mapping{\ensuremath{f}\xspace}
In Equation \ref{sec:suitable_ref_nodes}, we defined the set~\sel of all \emph{suitable reference nodes} including their \tthree witnesses.
We showed that $\refset(\sel)$ is exactly the set of nodes that \Rule1 targets.
Here, we now describe a simple and practical algorithm to compute \sel in time~$\Oh{n+m}$ with very small constants.
It consists of three parts:
\begin{enumerate}
    \item In \cref{alg:compute_supers}, we compute the superset ${\superS \supseteq \sel}$.
          It may also contain some ``erroneous'' \ttwo witnesses.
          Then, the following issue arises: since we cannot distinguish between \ttwo and \tthree nodes at this point,
          we may include reference nodes without any \tthree neighbors, \ie $\refset(\superS)$ can contain nodes for which \Rule1 does not apply.

    \item In~\cref{alg:filter_set}, we then remove $\isTypeAt{2}{\nodeu}{\refa}$ entries from \superS, where $\nodeu$ is a \ttwo neighbor to $\refa$.
          Afterwards, only reference nodes with at least one \tthree neighbor remain.
          To test whether a node is a \tthree neighbor, we only have to ensure that all its neighbors are of at least \ttwo.
          While conceptually simple, the required subset tests are too expensive if done na\"ively, and we need another trick to do it in linear time:

    \item To prevent redundant subset tests for some node \nodeu relative to different reference nodes, we first introduce a partial mapping~\mapping.
          We say that \mapping is \emph{properly partitioning} if it maps all neighbors of a \tthree witness $\isThreeAt{\nodeu}{\refa} \in \sel$ to its canonical reference node~$\rho$.
          This allows us to test the neighborhood inclusion at most once per node.
          We show that \mapping always exists and can be computed in linear time.
\end{enumerate}

\subsection{Computing \superS: a set of \ttwothree neighbors of canonical reference nodes}
\def\reps{\ensuremath{\mathsf{canRef^{-1}}}\xspace}
\begin{algorithm2e}[tb]
    Initialize Array~\reps with $\forall v \in V\colon\ \reps[v] \gets \emptyset$

    \BlankLine
    \tcp*[h]{\small Compute the inverse of~$\crepr$ as $\reps[\refa] = \setc{\nodeu \in V}{\nodeu \ne \refa \ \land\ \crepr(\nodeu) = \refa}$:}
    \AlgoDetails{
        \For{$\nodeu \in V$}{
            $\Delta \gets \max\setc{\deg(x)}{x \in N[\nodeu]}$\;
            $\crepr \gets \max\setc{\nodev \in N[\nodeu]}{\deg(\nodev) = \Delta}$\;
            \If{$\nodeu \neq \crepr$}{
                $\reps[\crepr] \gets \reps[\crepr] \cup \set{\nodeu}$\;
            }
        }
    }

    \medskip

    \tcp*[h]{\small Remove type 1 relations and compute $\superS$ using:}\label{line:blocktwo-start} 
    \AlgoDetails{
        $\superS \gets \emptyset$\;
        \For{$\refa \in V$}{
            Mark $\refa$ and all its neighbors with value \refa\;
            \For{$\nodeu \in \reps[\refa]$}{
                \If{\text{all neighbors of $\nodeu$ are marked by $\refa$}}{
                    $\superS \gets \superS \cup \set{\isTwoThreeAt{\nodeu}{\refa}}$\;
                }
            }
        }}\label{line:blocktwo-stop}

    \Return \superS\;

    \caption{
        Compute $\superS \supseteq \sel$.
        In contrast to $\sel$, the result~$\superS$ of this algorithm also contains \ttwo nodes as witnesses.
        They will be filtered out in \cref{alg:filter_set}.
    }
    \label{alg:compute_supers}
\end{algorithm2e}

The first step of our algorithm is outlined in \cref{alg:compute_supers}.
Its main idea is to implement our discussion on canonical reference nodes in \cref{subsec:canonical-reference-nodes}:
each node~$\nodeu$ selects the neighbor~$\refa$ with maximum value (with respect to the total ordering of $V$) among \nodeu's neighbors of maximum degree and attempts to become a witness to $\refa$.

The output of \cref{alg:compute_supers} is a list of all $\isTwoThreeAt{\nodeu}{\refa}$ where $\refa$ is the canonical reference node of $\nodeu$ and $\nodeu$ is a \ttwo or \tthree neighbor of~\refa.
Recall from the discussion in~\cref{subsec:canonical-reference-nodes}, that \sel must not contain \emph{all} $\isThreeAt{\nodeu}{\rho}$ relations, but only those where $\rho$ is the largest node (with respect to the total ordering of $V$) among its twins.
The same is true for the output of \cref{alg:compute_supers}.
The following lemma formalizes this notion:

\begin{lemma} \label{lemma:set}
    \cref{alg:compute_supers} outputs the set \superS, such that $\forall \isTwoThreeAt{\nodeu}{\refa} \in \superS\colon N(u)\subseteq N[\refa]$.
    Further $\superS \supseteq \sel$ contains at least all entries of $\sel$.\footnote{%
        Recall that the type~$T$ of $\isTypeAt{T}{\nodeu}{\refa}$ is only a hint to the reader.
        Hence, we say $\forall T_1, T_2\colon \isTypeAt{T_1}{\nodeu}{\refa} = \isTypeAt{T_2}{\nodeu}{\refb}$, which allows us to compare $\superS$ containing $\isTwoThreeAt{\cdot}{\cdot}$ and $\sel$ containing $\isThreeAt{\cdot}{\cdot}$.
    }
\end{lemma}

\begin{proof}
    The first claim is trivially established by the marking scheme in lines \ref{line:blocktwo-start} to \ref{line:blocktwo-stop} of \cref{alg:compute_supers}.
    Thus, we only focus on the second claim here.
    Assume $\exists \isThreeAt{\nodeu}{\refa} \in \sel$ for which $\isTwoThreeAt{\nodeu}{\refa} \notin \superS$.
    Then clearly \cref{alg:compute_supers} has selected for $\nodeu$ a reference node $\refb \neq \refa$ with $\deg(\refb) \geq \deg(\refa)$.
    But since $\nodeu$ is a \tthree neighbor of $\refa$ and a neighbor of $\refb$, we again have $\refb \in N_2(\refa) \cup N_3(\refa)$ and thus $N[\refb] \subseteq N[\refa]$ as well as $\deg(\refb) \leq \deg(\refa)$.
    Hence, $N[\refa] = N[\refb]$ implying that $\refa$ and $\refb$ are twins.
    Then, because the reference node is chosen as in the definition of $\sel$, we must have $\refa = \refb$, contradicting the initial assumption.
\end{proof}

The runtime of~\cref{alg:compute_supers} follows by observing that each node~$\nodeu \in V$ is contained in $\nodeu \in \reps(\refa)$ for at most one $\refa$.
Thus, the nested loop in line~12 iterates over all nodes at most once.
Then considering their neighbors visits each edge at most once:

\begin{observation}
    Alg. \ref{alg:compute_supers} runs in time ${\mathcal{O}(n+m)}$.
\end{observation}

\subsection{Filtering out \ttwo witnesses}
\begin{figure}
    \begin{center}
        \scalebox{1.0}{
        \begin{tikzpicture}[snode/.style={inner sep=0, minimum width=0.5em, draw, circle}]
            \node[draw, cloud, minimum width=7em, minimum height=6em] {$L$};
            \node[draw, cloud, minimum width=7em, minimum height=6em] at (10em, 0) {$R$};

            \foreach \y [count=\z] in {2,1,0,-1,-2} {
                    \foreach \yy in {2,1,0,-1,-2} {
                            \path[draw, black!50] (3.5em - \intcalcAbs{\y} * 0.5em, \y em)
                            to (7em + \intcalcAbs{\yy} * 0.5em, \yy em);
                        }
                }

            \node[snode, label=below:{\footnotesize $r_i$}] (ri) at (19em-7em, 0) {};

            \node[snode, label=above:{\footnotesize $a_i$}] (ai) at (22em-7em, 2em) {};
            \node[snode, label=below:{\footnotesize $b_i$}] (bi) at (22em-7em, 0) {};

            \node[snode] (a2) at (23em-7em, 2em) {};
            \node[snode] (a3) at (24em-7em, 2em) {};
            \node[snode] (a4) at (25em-7em, 2em) {};

            \path[draw] (ai) to (a2) to (a3) to (a4) to[bend right, looseness=2] (ai);
            \path[draw] (ri) to (ai) to (bi) to (ri);

        \end{tikzpicture}}
    \end{center}
    \caption{
        Super-linear runtime for the na\"ive \Rule1 for the pre-selected reference nodes $\refset(\superS)$.
        \textbf{Construction:} The two subgraphs~$L$ and $R$ are $k$-cliques on nodes $V_L = \set{\ell_1, \ldots, \ell_k}$ and $V_R = \set{r_1, \ldots, r_k}$, respectively.
        Each $\ell_i$ is connected to all nodes in $V_R$.
        Further, we attach to $r_i \in V_R$ its own copy of the gadget illustrated.
        Observe that
        (i) $a_i \in N_1(r_i)$,
        (ii) $b_i \in N_2(r_i)$, and
        (iii) that the whole gadget cannot be reduced by \Rule1.
        We assume that the graph is provided as an ordered adjacency list, where $\ell_i$ are listed before $r_j$ and $a_i, b_i$.
        \textbf{Analysis:} Let $k > 1$.
        Then, each gadget's~$b_i \in N_2(r_i)$ will pick~$r_i$ as its canonical reference node, \ie $\setc{\isTwoThreeAt{b_i}{r_i}}{1 \le i \le k} \subseteq \superS$ while the graph has no \tthree nodes at all (nodes in $L$ and $R$ will also pick a reference node in $R$).
        Now in order to partition the neighborhood $N(r_i)$ for a fixed $r_i \in R$, we have to inspect the neighborhood of each node in $V_L$ which takes time $\Theta(k^2)$.
        Thus, in total the na\"ive test requires $\Omega(k^3)=\Omega(n^3)$ time, as
        $n = \Theta(k)$ and $m=\Theta(k^2)$.
    }
    \label{fig:superlinear-runtime}
\end{figure}
In the previous section, we computed the superset $\superS \supseteq \sel$ in linear time.
What remains to do is to prune all ``false'' entries. 
To this end, we iterate over \superS and keep all $\isTwoThreeAt{\nodeu}{\refa}$ where $\nodeu$ is a \tthree neighbor to \refa.
In other words, we need to ensure for node~\nodeu that all neighbors $N[\nodeu] \setminus \set{\refa}$ are of types~$2$ or~$3$.
Unfortunately, \superS does not contain \emph{all} \ttwothree relations.
Thus we need to compute the relevant neighborhood types explicitly.
Even worse, as a node~\nodeu can appear in contexts of various reference nodes, a na\"ive implementation of this process exhibits super-linear runtime (see \cref{fig:superlinear-runtime}).

Instead, we use an ``oracle''~\mapping to tell us whether the type of node~$x$ is relevant, and if so, for which \emph{unique} reference node~$\mapping(x)$ it needs to be considered.
More formally, we ask for:

\begin{definition}\label{def:proper-partition}
    We say a partial mapping ${\mapping\colon V\to V}$ is properly partitioning with respect to~\sel if for all $\isThreeAt{\nodeu}{\refa} \in \sel$, and all $x \in (N[u]\setminus\{\rho\})$, we have $f(x) = \rho$.
\end{definition}

Observe that this property is defined in terms of the very output~\sel we want to ultimately compute.
Yet, we show in \cref{subsec:mapping} that we can obtain a proper~\mapping in linear time without knowing~\sel.
For exposition, let us assume that we already have \mapping available.

The oracle allows \cref{alg:filter_set} to establish the neighbor type of node~\nodeu at most once, namely in the context of $f(u)$.
Since each test costs time $\Oh{1 + \deg(\nodeu)}$ the total cost is bounded by $\Oh{n+m}$.
The intuition is as follows:
Consider a proper \tthree witness $\isThreeAt{\nodeu}{\refa} \in \sel$.
Then, the oracle ensures that all its neighbors are evaluated in the context of \refa{} --- and will be found as \ttwothree themselves.
If, on the contrary, a neighbor is considered for a different reference node~$\refb \ne \refa$ (or not at all), it is correct to implicitly treat \nodeu as \tone for \refa.
This is formalized in the following \cref{lem:filter_set_is_correct}:

\def\ttt{\ensuremath{T^{(\mapping)}_{\set{2,3}}}\xspace}
\def\invMapping{\ensuremath{\mathsf{chosenBy}}}
\def\output{\ensuremath{\mathcal{Y}}}
\begin{algorithm2e}[tb]
    Precompute $\invMapping(\refa) = \setc{x \in V}{f(x) = \refa}$, \ie $\mapping^{-1}$ by iterating once over $\mapping(v)$.

    \smallskip

    Initialize $\output = \emptyset$\;
    Initialize two marking slots per node\;
    \For{$\refa \in \refset(\superS)$}{
        Mark $\refa$ and its neighbors with value \refa\ in slot~1\;

        \For{$\nodeu \in \invMapping(\refa)$}{
            \If{all $N(\nodeu)$ have mark~$\refa$ in slot~1\label{alg:line_filter_set_cond}}{
                Mark \nodeu with value \refa in slot~2\label{alg:line_is_two_three}\;
            }
        }

        \For{$\nodeu \in \invMapping(\refa)$}{
            \If{all nodes $(N[\nodeu] \setminus \set{\refa})$ are marked by $\refa$ in slot~2}{
                $\output \gets \output \cup \set{\isThreeAt{\nodeu}{\refa}}$\label{alg:line_add_output}\;
            }
        }
    }

    \Return \output\;

    \caption{Filter \superS by removing \ttwo witnesses using a \emph{proper partition} function~\mapping.}
    \label{alg:filter_set}
\end{algorithm2e}

\begin{lemma}\label{lem:filter_set_is_correct}
    \Cref{alg:filter_set} outputs $\sel$.
\end{lemma}
\begin{proof}
    \def\output{\ensuremath{\mathcal Y}\xspace}
    Let~\output be the output of Alg. \ref{alg:filter_set} for parameters $\superS$ and \mapping.
    In the following, we prove that $\output = \sel$ in two steps.
    We show that (i) $\sel \subseteq \output$ and (ii) $\output \subseteq \sel$.

    \begin{itemize}
        \item Consider $\isThreeAt{\nodeu}{\refa} \in \sel$.
              By assumption, we have $\forall v \in N[u] \setminus \set{\refa}\colon f(v) = \refa$ and therefore $N[u] \setminus \set{\refa} \subseteq \invMapping(\refa)$.
              Since $u$ is a \tthree node to $\refa$, for every node $v \in N[u]$, we have $N[v] \subseteq N[\refa]$ and the first marking scheme thus identifies all nodes in $N[u] \setminus \set{\refa}$ as \ttwothree.
              In that case, it marks \nodeu a second time.
              Therefore, in the second iteration, because $u \in \invMapping(\refa)$ and all nodes $N[u] \setminus \set{\refa}$ are marked twice, we add $\isThreeAt{\nodeu}{\refa}$ to $\output$.

        \item Consider the relation $\isTwoThreeAt{\nodeu}{\refa} \in \output$.
              Then, $\isTwoThreeAt{\nodeu}{\refa}$ must have been added to $\output$ in line~\ref{alg:line_add_output}.
              This only happens if $\nodeu \in \invMapping(\refa)$ and all nodes $N[u] \setminus \set{\refa}$ were marked twice.
              This, in turn, implies that $N[u] \setminus \set{\refa} \subseteq \invMapping(\refa)$ and that $\forall v \in N[u]\colon N[v] \subseteq N[\refa]$; otherwise not all nodes in $N[u] \setminus \set{\refa}$ would been marked a second time in line~\ref{alg:line_is_two_three}.
              Hence, $\nodeu$ is a \tthree node for $\refa$ and since $f(\nodeu) = \refa$, by definition, $\isThreeAt{\nodeu}{\refa} \in \sel$.
    \end{itemize}
\end{proof}

\noindent
The runtime of \cref{alg:filter_set} can be bounded by observing that each node~$u \in V$ chooses at most one reference node in $\superS$.
Thus, over the course of the algorithm, the nested loops (in lines~6 and 9) iterate over a subset of $V$, resulting in a total of $\Oh{n}$ iterations.
In the same spirit, the conditions in lines~7 and~10 consider each edge at most once.
This leads to:

\begin{observation}
   Alg. \ref{alg:filter_set} runs in $\Oh{n+m}$ time.
\end{observation}

\subsection{Computing a proper partition}\label{subsec:mapping}
In order to apply \cref{alg:filter_set}, we first need to compute a partial mapping $f\colon V \to V$ that is properly partitioning with respect to~$\sel$.
For that, we need to map every \ttwo node that is neighbor to some \tthree node $u$ with $\isThreeAt{u}{\refa} \in \sel$ to $\refa$ in $f$.
At this point however, we do not know which nodes are \tthree nodes in $\witset(\superS)$ and neighbored \ttwo nodes are possibly neighbored to more than one reference node in $\refset(\superS)$.

Hence, to correctly identify these reference nodes, we make use of the fact that every node $u'$ in the neighborhood of $\refa$ that chose another canonical reference node $\refb \in \refset(\superS), \refb \neq \refa$ did so in \cref{alg:compute_supers} due to our tiebreaker because it is neighbored to both $\refa$ and $\refb$ (otherwise $u' \lhd_1 \sigma$).
Since $\isThreeAt{u}{\refa} \in \sel$, $u$ is not neighbored to $\refb$ (otherwise it would have also chosen $\refb$ as its canonical reference node).
Hence, a \ttwo node can correctly identify $\refa$ among its neighbored reference nodes by applying a reversed version of the tiebreaker in \cref{alg:compute_supers} to all its neighbored reference nodes in $\refset(\superS)$.

\begin{algorithm2e}[tb]
    Initialize array~$R$ of ref. nodes $\forall v \in V\colon\ R[v] = \infty$\;
    \For{$\isTwoThreeAt{\nodeu}{\refa} \in \superS$ \label{alg:secondloop}}{
        $R[\nodeu] \gets \rho$\tcp*{By \cref{alg:compute_supers}: each $\nodeu \in V$ is a witness at most once}
    }

    \smallskip

    Initialize empty partial mapping $\mapping\colon V \to V$\;
    \For{$x \in N[\witset(\superS)]$ \label{alg:thirdloop}}{
        \tcp{Map $x$ to reference node of smallest id among neighbors with smallest degree}
        $\mathsf{refs} \gets \setc{R[y]}{y \in N[x] \text{ with } R[y] \in N(x) }$\;
        \If(\hfill\texttt{// $N[x] \cap \refset(\superS) = \set{x}$}){$\mathsf{refs} = \emptyset$}{
            \textbf{continue};
        }
        $\mathsf{minDegInRefs} \gets \min\setc{\deg(y)}{y \in \mathsf{refs}}$\;
        $r \gets \min\setc{x \in \mathsf{refs}}{\deg(x) = \mathsf{minDegInRefs}}$\; \label{alg:lk4}
        \medskip
        Add mapping $\mapping[x] \gets r$\;
    }

    \Return{\mapping}\;

    \caption{Compute a partial mapping $\mapping\colon V \to V$ that is properly partitioning, \ie
        $\forall \isThreeAt{\nodeu}{\refa} \in \sel\colon \forall x \in (N[\nodeu]\setminus\{\refa\})\colon f(x) = \refa$.
    }
    \label{alg:mapping}
\end{algorithm2e}

\medskip

\begin{lemma}\label{lemma:construct}
    \Cref{alg:mapping} computes a partial mapping $\mapping\colon V {\to} V$ that properly partitions with respect to~\sel.
\end{lemma}
\begin{proof}[Proof (by contradiction)]
    Let $f$ be the output of \cref{alg:mapping}.
    Assume that there exists $\isThreeAt{u}{\refa} \in \sel$ such that $\exists x \in (N(u) \setminus \set{\refa})$ with $f(x) = r_x \neq \refa$.

    Since $\isThreeAt{u}{\refa} \in \superS$, we consider $x \in N[\witset(\superS)]$ in the loop and have $r_x \in \mathsf{refs}$ in that iteration because a neighbor $y \in N[x]$ proposed $R[y] = r_x$.
    But, by definition, $x$ is either a \ttwo or \tthree node to $\refa$ and thus $N[x] \subseteq N[\refa]$ implying $y \in N[\refa]$.

    Hence, $y$ is connected to both $\refa$ and $r_x$.
    Since $R[y] = r_x$, by \cref{alg:compute_supers}, we have
    \begin{equation*}
        \deg(r_x) > \deg(\refa) \quad\lor\quad \left(\deg(r_x) = \deg(\refa) \land r_x > \refa\right).
    \end{equation*}
    In either case, because $u \in N(x)$ and $R[u] = \refa$, we also have $\refa \in \mathsf{refs}$, and the tiebreaker would thus choose $f(x) = \refa$ and not $f(x) = r_x$.
\end{proof}

\begin{observation}
    Alg. \ref{alg:mapping} runs in time $\Oh{n+m}$.
\end{observation}

\section{Applying and Extending the Reductions}\label{sec:applying}
\def\refs{\ensuremath{\mathcal R}\xspace}
\begin{algorithm2e}[tb]
    \def\marked{\ensuremath{\mathsf{marked}}\xspace}
    \def\del{\ensuremath{\mathsf{deletable}}\xspace}

    Initialize set $\marked = \emptyset$\;
    \For{$\refa \in \refs$}{
        $\marked \gets \marked \cup N[\refa]$\label{alg:line-mark} \;
    }

    \smallskip

    Initialize $\del = \emptyset$\;
    \For{$\refa \in \refs$}{
        \For{$\nodeu \in (N(\refa) \setminus \refs)$} {
            \If{all neighbors of \nodeu are marked, \ie $N(\nodeu) \subseteq \marked$ \label{alg:line-marking-cond}}{
                $\del \gets \del \cup \set{\nodeu}$
            }
        }
    }

    \medskip
    Delete all nodes and incident edges in \del\;

    \caption{Apply \Rule1-style reductions to all reference nodes in $\refs$.}
    \label{alg:remove}
\end{algorithm2e}

In \cref{sec:linear_time_search} we showed how to compute the set~${\refs \coloneqq \refset(\sel)}$ of reference nodes to which \Rule1 can be applied to.
Now we connect~$\refs$ back to the problem of reducing a \Dom instance.
From the correctness~\cite{DBLP:journals/jacm/AlberFN04} of \Rule1, we know that there exists a minimal dominating set $\domset \supseteq \refs$.

Thus, it is safe to add all nodes of~$\refs$ to the dominating set.
Alber \etal~implicitly encode this, by adding a gadget leaf node to each reference node in~$\refs$, while practical solver implementations may choose to explicitly fix these nodes.
Since it is trivial to implement either way in linear time, we do not concern ourselves with the details:
for the remainder of this section, we only assume that all nodes of $\refs$ are added to the dominating set in some way.

Given that the nodes~$\refs$ are part of the dominating set, we want to reduce the \Dom instance~$G$ as much as possible.
\Cref{alg:remove} is a first step designed to be reasonably close to the original \Rule1-style reduction of \cref{def:rule1}.
The central idea is to \emph{mark} the closed neighborhood of each reference node in \refs.
Then, a marked node~\nodeu is either in the dominating set itself (\ie $\nodeu \in \refs$) or covered by a neighbor (\ie $N(\nodeu) \cap \refs \ne \emptyset$).
In any case, it is ``taken care of'' as the \Dom requirements towards the node have been satisfied;
yet, it may still be beneficial to add such a node, if it covers sufficiently many unmarked nodes.
The following nodes, however, can be deleted:

\begin{observation}\label{lem:correct-removal}
    Assuming nodes~\refs are added to the dominating set, all marked nodes whose entire neighborhood is marked are redundant (\ie deletable).
\end{observation}

\begin{figure}
    \begin{center}
        \scalebox{1.0}{
        \begin{tikzpicture}[
                vertex/.style={draw, circle, minimum width=1.6em, minimum height=1.6em, inner sep=0},
                edge/.style={draw},
            ]

            \node[vertex] (l1) at (0,0) {$x$};
            \node[vertex, fill=black!10] (r1) at (1.5,0) {$\refa$};
            \node[vertex] (m1) at (3,0) {\nodeu};
            \node[vertex] (m2) at (4.5,0) {\nodev};
            \node[vertex, fill=black!10] (r2) at (6,0) {$\refb$};
            \node[vertex] (l2) at (7.5,0) {$y$};

            \node[minimum width=7.5em, minimum height=3em, draw, opacity=0.3] at (3.75, 0) {};

            \path[edge] (l1) to (r1);
            \path[edge] (r1) to (m1);
            \path[edge] (m1) to (m2);
            \path[edge] (m2) to (r2);
            \path[edge] (r2) to (l2);
        \end{tikzpicture}}
    \end{center}
    \vspace{-1em}
    \caption{
        Deletion of \tone nodes.
        \Rule1 applies to $\refs = \set{\refa, \refb}$ as witnessed by $x \in N_3(\refa)$ and $y \in N_3(\refb)$.
        Observe that the middle nodes~$\nodeu$ and $\nodev$ are \tone neighbors to $\refa$ and $\refb$, respectively.
        Hence, the original \Rule1 leaves them untouched.
        In contrast, \cref{alg:remove} marks all nodes. Thus it also deletes the middle nodes.
    }
    \label{fig:addition-removals}
\end{figure}

The previous observation establishes that \cref{alg:remove} is safe.
In the following, we convince ourselves that we prune at least all nodes the original \Rule1 removes.
\begin{lemma}
    Applying \cref{alg:remove} to $\refs \coloneqq \refset(\sel)$ prunes at least all nodes that an exhaustive application of the original \Rule1 deletes.
\end{lemma}
\begin{proof}
    In \cref{lemma:maxcard}, we already showed that $\refs$ is the full set of nodes where \Rule1 is applicable.
    From \cref{lemma:unaffected}, we also know that the order in which we apply the reductions is irrelevant.
    Thus, we only need to show that we carry out the transformations correctly.

    Let $\refa \in \refs$ be a reference node.
    Then, \Rule1 deletes all its $(N_2(\refa) \cup N_3(\refa))$ neighbors.
    Recall, that a neighbor is of \ttwothree for the reference node~\refa, if all its neighbors are contained in $N[\refa]$:
    \begin{equation}
        \forall \nodeu {\in} N(\refa)\colon \
        \nodeu \in N_2(\refa) {\cup} N_3(\refa) \ \Leftrightarrow\ N[\nodeu] {\subseteq} N[\refa]\,.
    \end{equation}

    The marking scheme of \Cref{alg:remove} implements the second condition:
    If a reference node~$\refa$ marks all its neighbors $N[\refa]$, then the entire closed neighborhood of some $\nodeu \in N(\refs)$ is marked if $N[u] \subseteq N[\refa]$.
    Thus, we identify all nodes that the original rule deletes.
\end{proof}

\begin{observation}
    Alg. \ref{alg:remove} runs in time $\Oh{n + m}$.
\end{observation}

Thus, we know that we are doing at least as well as the original formulation.
In fact, as illustrated in \cref{fig:addition-removals}, we are already doing better.

\subsection{Pruning even more nodes}\label{subsec:prune-more-nodes}
\begin{figure}
    \begin{center}
        \scalebox{1.0}{
        \begin{tikzpicture}[
                vertex/.style={draw, circle, minimum width=1.6em, minimum height=1.6em, inner sep=0},
                edge/.style={draw},
            ]

            \node[vertex] (l1) at (0,0) {$x$};
            \node[vertex, fill=black!10] (r1) at (1.2,0) {$\refa$};
            \node[vertex] (m1) at (2.4,0) {\nodeu};
            \node[vertex] (c) at (3.6, 0) {$c$};
            \node[vertex] (m2) at (4.8,0) {\nodev};
            \node[vertex, fill=black!10] (r2) at (6,0) {$\refb$};
            \node[vertex] (l2) at (7.2,0) {$y$};

            \node[minimum width=3em, minimum height=3em, draw, opacity=0.3, label=below:{\footnotesize only unmarked node}] at (c) {};

            \path[edge] (l1) to (r1);
            \path[edge] (r1) to (m1);
            \path[edge] (m1) to (c);
            \path[edge] (c) to (m2);
            \path[edge] (m2) to (r2);
            \path[edge] (r2) to (l2);
        \end{tikzpicture}}
        \vspace{-1em}
    \end{center}
    \caption{
        Deletion of more \tone nodes.
        \Rule1 applies to $\refs = \set{\refa, \refb}$ as witnessed by $x \in N_3(\refa)$ and $y \in N_3(\refb)$.
        Then, \cref{alg:remove} marks all nodes but the middle~$c$ which prevents the deletion of nodes~$\nodeu$ and~$\nodev$ (since they are adjacent to the unmarked $c$).
    }
    \label{fig:addition-removals-plus}
\end{figure}

The pruning scheme (Algorithm~\ref{alg:remove}) and its underlying correctness guarantees in~\cref{lem:correct-removal} are designed to mirror the original \Rule1.
Here we state a stronger version:

\begin{lemma}
    Assuming nodes~\refs are added to the dominating set, all marked nodes~$\nodeu \notin \refs$ with \emph{at most one unmarked} neighbor are redundant.
\end{lemma}
\begin{proof}
    \cref{lem:correct-removal} establishes the claim for marked nodes without unmarked neighbors.
    Thus, we consider only the remaining case:
    Let $\nodeu$ be a marked node with a single unmarked neighbor~$x \in N(\nodeu)$.
    Clearly, $\nodeu \notin \refs$ as all nodes in~\refs are added to the dominating set and thus mark their entire neighborhood.

    The claim follows from an exchange argument analogously to \Rule{Leaf}.
    Assume $\mathcal{D}$ is a minimum dominating set containing $\refs \cup \set{\nodeu}$.
    Removing node~\nodeu from $\domset$ only leaves node~$x$ uncovered since all other nodes are marked.
    Hence, we obtain another dominating set $\domset' = \domset \setminus \set{\nodeu} \cup \set{x}$ of the same cardinality.
\end{proof}

We can implement this stronger variant by modifying line~\ref{alg:line-marking-cond} in \cref{alg:remove} in the obvious way without affecting the asymptotic runtime.
As exemplified in \cref{fig:addition-removals-plus}, the modification leads to strictly stronger pruning capabilities.
In this case, even multiple rounds of the reduction rule may become possible.

\subsection{Pruning more edges}\label{subsec:prune-more-edges}
While the marking scheme of \cref{alg:remove} is designed to remove as many nodes as possible, some \tone neighbors can remain.
However, we can remove edges between those nodes:

\begin{lemma}
    Assuming nodes~\refs are added to the dominating set, all edges $\set{x, y}$ between marked nodes~$x \notin \refs$ and~$y \notin \refs$ are redundant.
\end{lemma}
\begin{proof}
    Let $\set{x,y} \in E$ be an edge between two marked nodes $x, y \notin \refs$.
    Since both endpoints are marked they are already covered through \refs.
    Thus, there is no ``internal'' need to add either $x$ or $y$ to the dominating set.
    While the nodes may be relevant for potential unmarked neighbors $N(x)$ and $N(y)$, respectively, adding $x$ or $y$ (or both) to the dominating set does not affect the coverage of the other.
\end{proof}

The lemma can be implemented in time $\Oh{n+m}$ by iterating over all edges remaining at the end of \cref{alg:remove} and checking the condition in constant time per edge.

\section{Experiments}\label{sec:experiments}
In this section, we evaluate the various \Rule1 variants discussed in this paper on real-world and synthetic datasets.
If not stated otherwise, we performed the experiments on three datasets denoted as \textsc{Main-D}; two derived from publicly available instances of recent PACE iterations, namely the heuristic instances from the PACE~2023 \textsc{Twin-Width} and the PACE~2025 \Dom tracks.
And thirdly, real-world instances from \textsc{Network Repository}~\cite{DBLP:journals/sigkdd/RossiA15}, a dataset comprised of graphs from various domains, among them are social, biological, brain and road networks.
Since some instances from \textsc{Network Repository} are directed networks, we treat each instance as an undirected graph --- multi-edges arising from this circumstance are removed.
We focus on large graphs, since these were often deemed out of reach for a faithful implementation of \Rule1, hence we only consider instances from \textsc{Network Repository} with more than $10^4$ edges.

All experiments are run on a server with an AMD EPYC 7702P 64-Core Processor and 512GB of RAM.
Our implementations are functionally equivalent to the formulations discussed in \cref{sec:applying}, reasonably optimized, and use the fast and low-level programming language Rust.
We consider the original \Rule1 (\textsc{Naive}), our new base algorithm sketched in \cref{alg:remove} (\textsc{Linear}), the extension of \cref{subsec:prune-more-nodes} (\textsc{Plus}), and finally the strongest version of~\cref{subsec:prune-more-edges} (\textsc{Extra}).
In our implementation, we also remove all incident edges around nodes in $\mathcal{D}$ after marking their entire neighborhood as \emph{covered}.

We implement \textsc{Naive} more efficiently than stated in the original paper by invoking \cref{alg:remove} for $\refs = \set{\nodeu}$ for every $\nodeu \in V$ only once (as opposed to multiple \Rule1 invocations).
The resulting algorithm has a runtime of $\Theta(n\cdot m)$.

\paragraph{Outline.} 
We split the experimental part of the paper into two parts: runtime and data reduction.
The former supports the claim that our new algorithms are not only more efficient in theory but also faster in practice.
Then, we quantify the effect of our improved rules on different aspects of data reduction quality.

\subsection{Runtime}

\begin{figure}[t]
    \centering
    \includegraphics[width=0.5\textwidth]{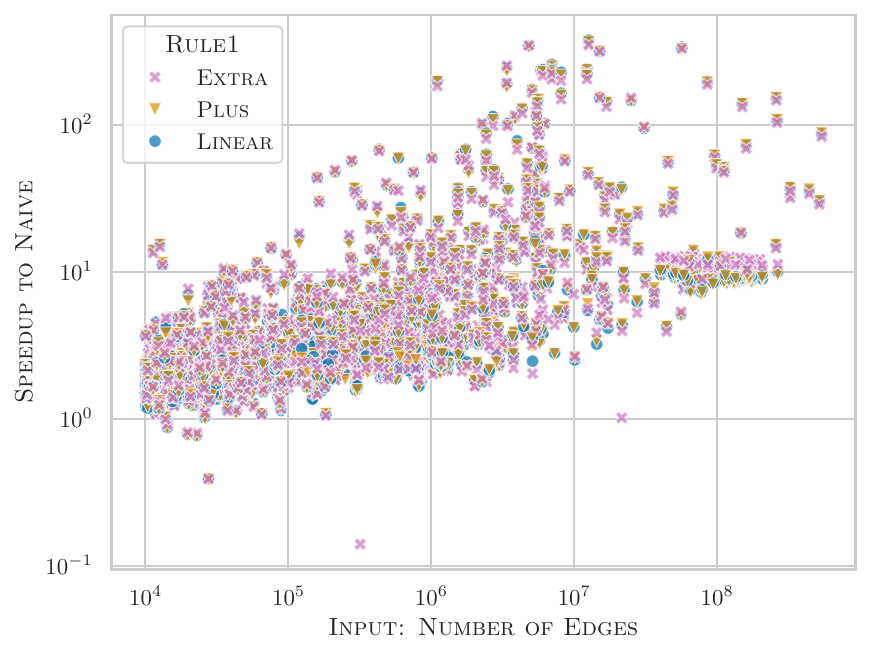}
    \caption{Speedup of various \Rule 1 variants over \textsc{Naive} evaluated on \textsc{Main-D}.}
    \label{fig:speedup}
\end{figure}

\begin{figure}[t]
    \centering
    \includegraphics[width=0.5\textwidth]{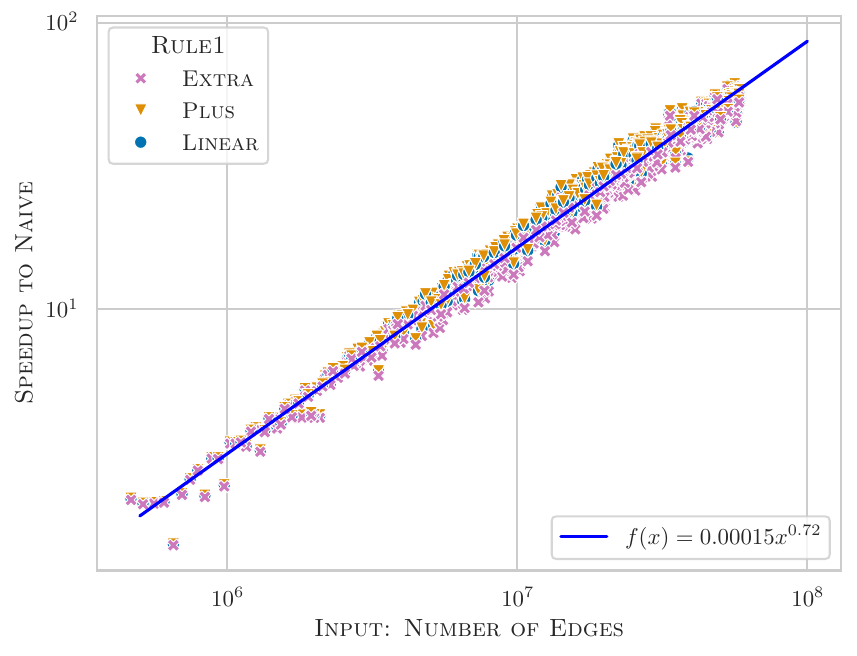}
    \caption{Speedup of \Rule1-variants over \textsc{Naive} for hyperbolic graphs with $\beta = 3.0$, temperature $T = 0$ and $n = 10^5$ nodes with varying average degrees.}
    \label{fig:scaling_time}
\end{figure}

\begin{figure*}[t]
    \centering
    \includegraphics[width=0.333\textwidth]{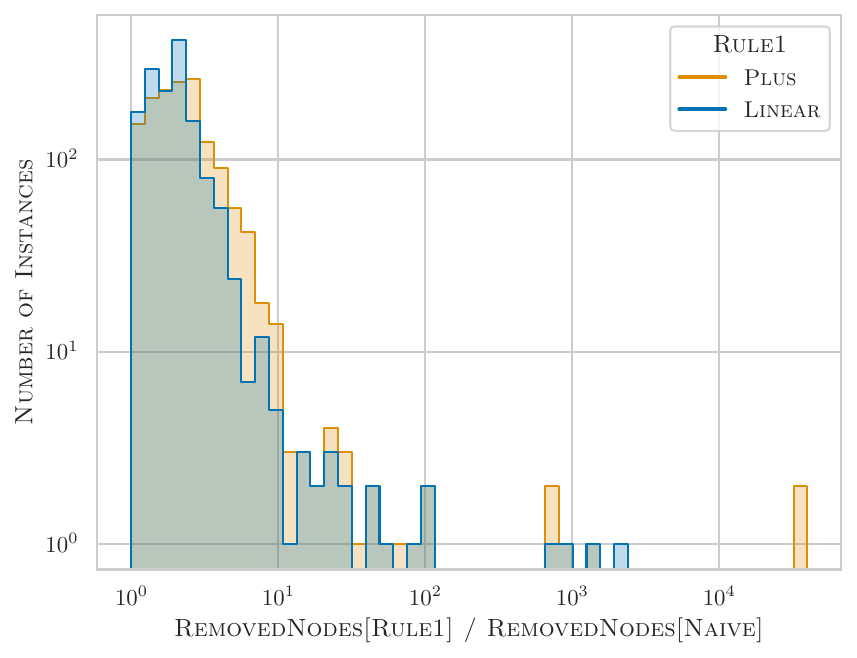}\hfill
    \includegraphics[width=0.333\textwidth]{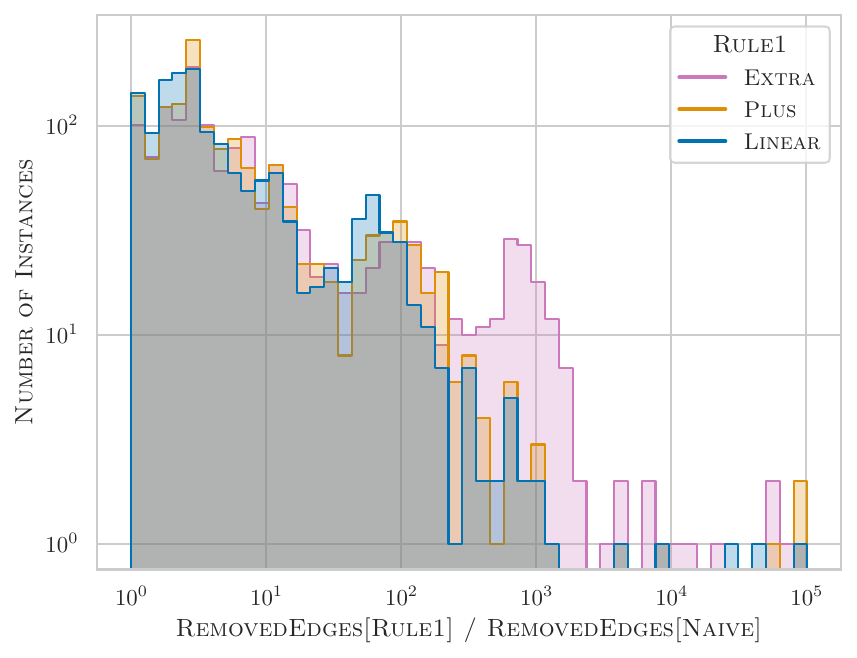}\hfill
    \includegraphics[width=0.333\textwidth]{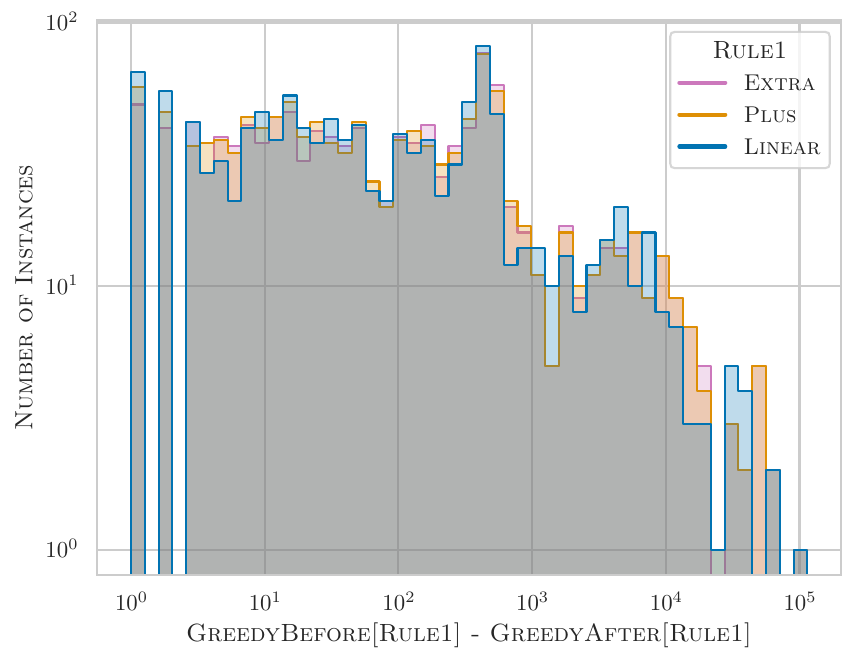}
    \caption{
        Relative increase in removed nodes \textit{(l)} / removed edges \textit{(m)} / total difference in greedy \Dom solution size \textit{(r)} in relation to the number of instances where this happened, evaluated on datasets \textsc{Main-D}.
        }
    \label{fig:pruning}
\end{figure*}
As evident from \cref{fig:speedup}, we obtain a speedup over \textsc{Naive} for all variants of our algorithm on almost every instance.
There are very few exceptions where \textsc{Naive} is actually faster; however, these are small instances and the maximum total runtime difference is only $0.23$ seconds.
In contrast, there are multiple instances where all variants of our algorithm are more than $100\times$ faster than \textsc{Naive}.
There are $12$ instances where \textsc{Naive} did not finish within the time budget of one hour, whereas every other variant finished within $58$ seconds on average. 

Across the whole dataset, we see a speedup of $12.2\times$, $12.0\times$, and $12.1\times$ for \textsc{Linear}, \textsc{Plus}, and \textsc{Extra} over \textsc{Naive} on average.
If we restrict the experiment to instances where $m\geq 10^6$, these speedups more than double to $26.5\times$, $26.1\times$, and $26.0\times$ for \textsc{Linear}, \textsc{Plus}, and \textsc{Extra}.

Although we observe a much better scaling behavior from our algorithms than \textsc{Naive}, the speedups are less pronounced than expected from such large asymptotic worst-case gaps.
Hence, let us quickly recall the runtime complexity of \textsc{Naive}, which is trivially bounded by $\Oh{\sum_{v\in V}\sum_{u \in N[v]} \deg(u)}=\Oh{n \cdot m}$.
Our experiments suggest that this bound is not representative of the algorithmic behavior of \textsc{Naive} for most instances.

Let us indulge in a little thought experiment.
Assume that the adjacency list of each node is randomly permuted; then the expected runtime of \textsc{Naive} is bounded from above by $\Oh{\mathbb{E}\left[\sum_{v \in V} \sum_{u \in N[v]} \operatorname{Geom}(p_{u,v})\right]}$, where $\operatorname{Geom}(p_{u,v})$ draws from the geometric distribution with success probability $p_{u,v} = 1-\frac{|N[u]\cap N[v]|}{|N[u]|+1}$.
Observe, we can abort the local neighborhood inspection of $u$ as soon as we find the first node in $N[u]\setminus N[v]$; the distribution of this time point is given by $\operatorname{Geom}(p_{u,v})$.
Clearly, this translates to an expected linear runtime for \textsc{Naive} on $\mathcal{G}(n,p)$ graphs\footnote{By $\mathcal{G}(n,p)$ graphs we refer to graphs of order $n$ drawn uniform at random from the Gilbert model~\cite{gilbert1959random}, where each edge exists independently with probability $p$.} for any $p\leq c$ where $0 \leq c < 1$ is any real valued constant.
This explains the two properties witnessed in \cref{fig:speedup}: the high amount of noise as well as the comparatively low speedups given the asymptotic gap between the algorithms.
Even under these circumstances, our algorithm outperforms \textsc{Naive} significantly on most instances.

Based on the observations in \cref{fig:speedup} and our thought experiment, we suggest that the runtime of \textsc{Naive} (and in extension the speedup of our algorithms) depends mainly on two parameters: namely, the neighborhood homogeneity and the average degree.
Since real-world networks are complex and exhibit a multitude of structural phenomena related to their specific evolution and domain, we use a random graph model to investigate the speedup of our algorithms to better control for the aforementioned parameters.

To this end, we generate hyperbolic random graphs with $n=10^5$ nodes and varying average degree using the generator of~\cite{DBLP:journals/netsci/BlasiusFKMPW22}.
Our choice of parameters induces graphs with a high clustering coefficient and a degree distribution following a power law, \ie $\mathbb{P}[\deg(v) = k] \sim k^{-\beta}$ for $\beta=3.0$.
Here, the high clustering coefficient is a desirable property for two reasons: it is an often observed property of real-world networks, and it can be seen as a proxy for the neighborhood homogeneity.

Observe that \cref{fig:scaling_time} supports our hypothesis that a larger average degree in combination with a high clustering coefficient translates into larger speedups of our algorithms over \textsc{Naive}.
As a result, the speedup can be described as a function in the number of edges --- increasing the edges by a factor of ten increases the speedup by at least a factor of $5$ on these instances.
Therefore, our model suggests that a high clustering coefficient in graphs serves as a magnifier that exacerbates the impact of the increasing number of edges significantly.

\paragraph{Multiple applications.}
Recall that \cref{lemma:unaffected} states that a single application of \Rule 1 (\ie \textsc{Naive} and \textsc{Linear}) suffices to find all possible \tthree nodes.
However, the extensions \textsc{Plus} and \textsc{Extra} allow the iterative application of these rules since additional node-removals can cut off entire parts (neighborhoods) of the graph.
We can further improve subsequent applications of \Rule1 by leveraging non-removed nodes that are marked as \emph{covered} (since they neighbored a now removed dominating node).
Then, we only consider uncovered neighbors of nodes for the type $1$ classification. 
Notice that multiple applications potentiate the speedups gained by our algorithms, as these are interdependent invocations of \Rule 1.
We only consider iterated applications of \textsc{Extra}, since \textsc{Plus} and \textsc{Extra} have very similar runtimes and the additional edge-removals of \textsc{Extra} have no influence on further \Rule1 applications.
In some cases, more than $1000$ iterations are possible until \textsc{Extra} no longer modifies the graph.
However, more than $95\%$ of instances terminate after at most 17 rounds.

\subsection{Data Reduction}
We consider mainly two dimensions of data reduction: (i) the number of nodes and edges that are removed by an application of the variants of \Rule 1, (ii) the improvement of the \textsc{Greedy} algorithm on the reduced instances.
In all our implementations (including \textsc{Naive}), we add a suitable reference node $\rho$ to the \Dom and remove $N_2(\rho)\cup N_3(\rho)$ without introducing a gadget node.

\paragraph{Pruning.}
\Cref{fig:pruning} indicates that \textsc{Plus} removes between $1{\times}-10{\times}$ more nodes than \textsc{Naive}; in four cases, we even observe up to $\num{10000}\times$ more removed nodes.
On average, \textsc{Plus} removes $59.8\times$ and \textsc{Linear} removes $6.1\times$ more nodes than \textsc{Naive}.
Even in the median, \textsc{Plus} removes $2.16\times$ more nodes than \textsc{Naive}.
In addition, \textsc{Linear}, \textsc{Plus} and \textsc{Extra} reduce on average $136.6\times$, $206.7\times$ and $410.9\times$ more edges than \textsc{Naive} with outliers as high as $10^5\times$ in some cases.
In the median, \textsc{Extra} still removes $4.95\times$ more edges than \textsc{Naive}.

The extremely high outliers can be attributed to graphs with star-like subgraphs where \textsc{Naive} does identify the center as a reference node but fails to remove most of the type $1$ `satellites' due to having no information about other reference nodes.
For our stronger variants, most of the `satellites' (and their edges to the reference node) are removable due to the extensions of \textsc{Linear} and \textsc{Plus}.

\paragraph{Effect on Greedy.}
Especially on massive graphs, which often permit only limited solvers due to resource and time constraints, there is a lack of powerful and safe reduction rules.
We study the effect of our reduction rules on the solution size obtained with one of the arguably fastest and simplest heuristics for \Dom, namely the \textsc{Greedy} heuristic.
To reduce the noise in the \textsc{Greedy} heuristic, resulting from tiebreaks between nodes of equal merit, we derive a consistent but random tiebreaker. 
We then run \textsc{Greedy} $10$ times, each time with a different random tiebreaker, and report the best solution size among the $10$ invocations.

The reduction rules improve the solution size of the \textsc{Greedy} heuristic by up to  $10^5$ nodes in absolute terms.
Upfront, this is in no way conspicuous, the interactions between \Rule1 and Greedy are non-trivial:
In \cref{sec:applying}, we prove that any potential \tthree node can only be a \tthree node to its neighbor of highest degree.
\textsc{Greedy}, on the other hand, iteratively adds a node with the most uncovered neighbors to be included in the \Dom.
Despite the obvious similarity, neither \Rule 1 nor its variants are subsumed by \textsc{Greedy}, and its application before \textsc{Greedy} results in (i) better solution sizes and (ii) only a negligible increase in runtime.

\section{Conclusion and Outlook}
We considered the established reduction rule, \Rule1 by Alber~\etal~\cite{DBLP:journals/jacm/AlberFN04} and improved both its asymptotic and practical speed, as well as its effectiveness.
As a common building block of kernelization algorithms/preprocessing routines, our contribution is expected to have a significant impact in practice.
Our most capable extension removes, on average, $59.8\times$ more nodes and $410.9\times$ more edges compared to the original formulation while being $12.1\times$ faster.
It would be interesting to see if one can improve the runtime complexity of higher order variants/generalizations as well~\cite{DBLP:journals/jacm/AlberFN04,Alber2006b}.
We leave these questions open for future work.

\section*{Acknowledgements}
The original research question arose while preparing for the Parameterized Algorithms and Computational Experiments (PACE) Challenge 2025 on solving Dominating Set in practice.

\clearpage

\bibliographystyle{siamplain}
\bibliography{article-bib2doi-checked} 

\end{document}